\documentclass[11pt,a4paper]{article}

\usepackage[utf8]{inputenc}
\usepackage[T1]{fontenc}
\usepackage[margin=2.5cm]{geometry}
\usepackage{graphicx}
\usepackage{caption}
\usepackage{dsfont}
\usepackage{algorithm, algorithmic}
\usepackage{amsmath,amssymb,amsthm}
\usepackage{tikz}
\usetikzlibrary{matrix, calc}
\usepackage[colorlinks=true,linkcolor=blue,citecolor=blue,urlcolor=blue]{hyperref}

\newcommand*\diff{\mathop{}\!\mathrm{d}}

\theoremstyle{plain}
\newtheorem{theorem}{Theorem}[section]
\newtheorem{lemma}[theorem]{Lemma}
\newtheorem{proposition}[theorem]{Proposition}

\theoremstyle{definition}
\newtheorem{definition}[theorem]{Definition}
\newtheorem{problem}[theorem]{Problem}
\theoremstyle{remark}
\newtheorem{remark}[theorem]{Remark}

\title{Maximizing Rényi Entropy over the Transportation Polytope:\\ Structure and Algorithm}
\author{Pierre Bertrand\\[2pt]
  \small Aix Marseille Univ, CNRS, AMSE, Marseille, France\\
  \small\texttt{pierre.bertrand.1@univ-amu.fr}}
\date{\today}

\begin{document}

\maketitle
\begin{abstract}
We study the problem of finding the minimum $\mathbb{L}_2$-norm element in a transportation polytope. It corresponds to minimizing the index of coincidence or maximizing Rényi entropy over joint probability distributions with prescribed marginals.
The solution is straightforward under a feasibility condition on the marginals.
We make three contributions.
First, we decompose the basic solution into a mixture of three interpretable couplings and derive a closed-form sampling procedure, revealing its role as the $\mathbb{L}_2$-projection of the uniform distribution onto the transportation polytope.
Second, we show that the feasibility condition becomes exponentially rare as the number of modalities grows, which highlights the practical necessity of a more general construction.
Third, we provide a KKT-based iterative construction of the optimal coupling for arbitrary marginals with explicit Lagrange multipliers.
\end{abstract}

\medskip
\noindent\textbf{Keywords:} Rényi entropy $\cdot$ Index of coincidence $\cdot$ Indeterminacy $\cdot$ Coupling $\cdot$ Transportation polytope $\cdot$ Marginal constraints

\smallskip
\noindent\textbf{Mathematics Subject Classification (2020):} 90C20 $\cdot$ 94A17 $\cdot$ 60E05

\section{Introduction\label{sec:intro}}

Let $X$ and $Y$ be discrete random variables taking values in $\{1,\dots,p\}$ and $\{1,\dots,q\}$, with prescribed marginals $\mu\in S_p$ and $\nu\in S_q$. We study the coupling problem of constructing $\pi\in\mathcal{S}_{\mu,\nu}$ that maximizes dispersion while preserving the marginals; such marginal-constrained coupling problems are classical in probability~\cite{FRE51,strassen1965existence} with applications in information theory~\cite{cicalese2019minimum}. Here
\[
\pi\in\mathcal{S}_{\mu,\nu}
\iff
\left\{
    \begin{array}{ll}
        \forall~1\le u \le p, &\pi_{u,\cdot} := \sum_{v=1}^q \pi_{u,v} = \mu_u, \\
        \forall~1\le v \le q, &\pi_{\cdot,v} := \sum_{u=1}^p \pi_{u,v} = \nu_v, \\
		\forall~1\le u \le p, \forall~1\le v \le q, &\pi_{u,v} \geq 0, \\
        &\pi_{\cdot,\cdot} := \sum_{u=1}^p \sum_{v=1}^q \pi_{u,v} = 1.
    \end{array}
\right.
\]
Throughout the paper, we assume without loss of generality that the marginals are sorted in non-decreasing order with no zero component:
\begin{equation}\label{eq:ordre_marges}
0 < \mu_1 \le \mu_2 \le \ldots \le \mu_p
\quad \& \quad
0 < \nu_1 \le \nu_2 \le \ldots \le \nu_q.
\end{equation}
For $\alpha>0$, $\alpha\neq 1$, the Rényi entropy of order $\alpha$ defined in \cite{renyi1961measures} is given by
\begin{equation}
H_{\alpha}(\pi) = \frac{1}{1-\alpha}\log\left(\sum_{u=1}^p\sum_{v=1}^q \pi_{u,v}^{\alpha}\right).
\end{equation}
We focus on $\alpha=2$, for which maximizing $H_2$ is equivalent to minimizing the \emph{index of coincidence}~(\cite{friedman1987index}):
\[
IC(\pi) := \sum_{u=1}^p\sum_{v=1}^q \pi_{u,v}^2.
\]
via $H_2(\pi)=-\log(IC(\pi))$. The IC measures the concentration of $\pi$; minimizing it yields the most dispersed coupling compatible with the marginals. Applications include cryptanalysis~\cite{friedman1987index,handbookcrypto} and market concentration measurement~\cite{hirschman1964paternity}. Our problem is
\begin{problem}[Maximum Rényi entropy / minimum index of coincidence]
\label{pb:max_entropy}
\[
\operatorname*{arg\,max}_{\pi\in\mathcal{S}_{\mu,\nu}} H_{2}(\pi)
\;=\;
\operatorname*{arg\,min}_{\pi\in\mathcal{S}_{\mu,\nu}} IC(\pi).
\]
The optimizer is unique (see Remark~\ref{rem:unique}) and is denoted $\pi^*$.
\end{problem}

Related work in information theory or cryptanalysis predominantly addresses the opposite optimization problem (minimum-entropy couplings); see, e.g.,~\cite{cicalese2019minimum} and algorithmic approaches such as~\cite{li2021efficient,kocaoglu2020applications}.

When Shannon entropy is used instead of $H_2$, the maximum-entropy solution is independence coupling~\cite{kullback_contingency_1968}. For $H_2$, simple calculations yield a closed-form additive solution (indeterminacy $\pi^+$ of Definition~\ref{def:pi_plus}) when the feasibility condition
\begin{equation}
\label{cond:H}
p\mu_1 + q\nu_1 \geq 1,
\end{equation}
is satisfied (see~\cite{pougaza2010maximum,bubak2025fitting} for related continuous constructions). Otherwise, the additive construction produces negative values (see~Proposition~\ref{prop:solution_pi_plus}).
We show in Section~\ref{sec:margins} that Condition~\eqref{cond:H} is rarely satisfied for generic marginals, which motivates an explicit solution in the general case.

\paragraph{Relation to the transportation polytope literature.}
Problem~\ref{pb:max_entropy} is equivalent to finding the minimum $\mathbb{L}_2$-norm element in the transportation polytope $\mathcal{S}_{\mu,\nu}$~\cite{bachem1980minimum,peyre2019computational}. Related work on this polytope typically considers a linear cost function~\cite{schumacher2018multiobjective}; Problem~\ref{pb:max_entropy} uses instead a quadratic cost, for which general-purpose QP solvers~\cite{wu2026quadratic} yield numerical approximations without an exact closed-form solution.
Romero~\cite{romero1990easy} first solved the relaxed problem---dropping the nonnegativity constraint---over real matrices, obtaining a closed-form solution that coincides exactly with indeterminacy $\pi^+$ (Definition~\ref{def:pi_plus}).
Calvillo and Romero~\cite{calvillo2016closest} subsequently solved the full non-negative problem, establishing the staircase zero structure of the optimizer and a row-by-row iterative algorithm of complexity $O(pq)$; their geometric approach, based on orthogonal projections onto linear varieties, does not require Lagrange multipliers and provides no dual certificate; nor is the question of when $\pi^+$ is itself optimal addressed.

\paragraph{Contributions.}
The present paper makes three contributions, each addressing a question left open by~\cite{romero1990easy,calvillo2016closest}:
\begin{enumerate}
  \item \textbf{Tightness of the relaxation.} We characterize exactly when $\pi^+=\pi^*$, a question not addressed in~\cite{romero1990easy,calvillo2016closest}: the proportion of marginal pairs satisfying Condition~\eqref{cond:H} equals $\frac{(p-1)!(q-1)!}{(p+q-2)!}$, which decays exponentially with $p$ and $q$, confirming that the general algorithm is needed in all but degenerate cases (Section~\ref{sec:margins}).
  \item \textbf{Explicit Lagrange multipliers.} We derive closed-form entrywise expressions for the KKT multipliers $r_{u,v}$ of the optimizer $\pi^*$, providing a complete dual certificate absent from~\cite{calvillo2016closest}, whose geometric approach does not require Lagrange multipliers (Sections~\ref{sec:forme}--\ref{sec:solution_generale}).
  \item \textbf{Structure of the relaxed solution $\pi^+$.} We decompose the indeterminacy coupling $\pi^+$ into three interpretable components, identify it as the $\mathbb{L}_2$-projection of the uniform distribution onto $\mathcal{S}_{\mu,\nu}$, and derive a closed-form sampling procedure (Section~\ref{ssec:decomposition}).
\end{enumerate}

\paragraph{Organization.}
Section~\ref{sec:solution_partielle} introduces indeterminacy $\pi^+$, establishes its optimality under Condition~\eqref{cond:H}, and derives its sampling decomposition. Section~\ref{sec:forme} uses KKT conditions to characterize the staircase zero structure of the general optimizer. Section~\ref{sec:hyp_rectangle} treats the rectangle case in closed form, providing the building block for Section~\ref{sec:solution_generale}, which presents the iterative algorithm and its complexity. Section~\ref{sec:margins} quantifies the restrictiveness of Condition~\eqref{cond:H}.

\section{Indeterminacy: Additive-Form Solution\label{sec:solution_partielle}}

\begin{remark}[Existence and uniqueness]\label{rem:unique}
Problem~\ref{pb:max_entropy} admits a unique optimizer~$\pi^*$. Existence follows from compactness of $\mathcal{S}_{\mu,\nu}$ and continuity of $IC$; uniqueness from strict convexity of $IC(\pi)=\sum_{u,v}\pi_{u,v}^2$ (Hessian $2I$, positive definite).
\end{remark}

The optimizer $\pi^*$ is characterized by the Karush--Kuhn--Tucker (KKT) conditions.

\begin{proposition}\label{prop:kkt}
The KKT conditions are necessary and sufficient for $\pi^*$.
\end{proposition}

\begin{proof}
Slater's condition holds: all constraints of $\mathcal{S}_{\mu,\nu}$ are affine, and the independence coupling $\pi^{\times}$ is a strict interior point (since all marginals are positive by~\eqref{eq:ordre_marges}). Hence any optimizer has associated Lagrange multipliers satisfying the KKT system. Sufficiency follows from strict convexity of $IC$ (Remark~\ref{rem:unique}) and linearity of all constraints.
\end{proof}

\subsection{Definition of Indeterminacy}
For Shannon entropy, solving the KKT system gives independence~\cite{kullback_contingency_1968}. For $H_2$, the KKT system without the positivity constraint yields the following closed form (see~\cite{romero1990easy,PBThese}), which is a solution only when it happens to be non-negative.

\begin{definition}[Indeterminacy]
\label{def:pi_plus}
\begin{equation}
\pi^{+}_{u,v} = \frac{\mu_u}{q} + \frac{\nu_v}{p} - \frac{1}{pq}, \quad \forall\, 1\le u \le p,\ \forall\, 1\le v \le q.
\end{equation}
\end{definition}

Indeterminacy was introduced in~\cite{MAR84,MAR86} within Mathematical Relational Analysis. The "+" notation mirrors the additive form, by analogy with the product form of independence $\pi^{\times}$.

\begin{remark}
The coupling $\pi^{+}$ also appears in the definition of the Janson--Vegelius criterion $JV$~\cite{JV77}, which can be seen as a cosine similarity to $\pi^{+}$ (as opposed to the $\chi^2$ criterion, based on independence $\pi^{\times}$).
\end{remark}

The coupling $\pi^+$ satisfies all equality constraints of $\mathcal{S}_{\mu,\nu}$, but nonnegativity depends on the marginals as stated in Proposition~\ref{prop:solution_pi_plus}.

\begin{proposition}[$\pi^+$ as a solution]\label{prop:solution_pi_plus}
The equality $\pi^* = \pi^+$ holds if and only if the Condition~\eqref{cond:H} holds.
\end{proposition}

\begin{proof}
By Proposition~\ref{prop:kkt}, $\pi^*$ is characterized by its KKT conditions. Relaxing nonnegativity, the KKT system yields $\pi^* = \pi^+$ (see~\cite{PBThese}). Restoring nonnegativity, $\pi^+$ is feasible if and only if Condition~\eqref{cond:H} holds.
\end{proof}

When Condition~\eqref{cond:H} fails, $\pi^+$ has negative entries and is not a valid coupling; it nonetheless serves as the starting point for the general construction.

\subsection{Decomposition of Indeterminacy\label{ssec:decomposition}}

Definition~\ref{def:pi_plus} can be rewritten as a mixture of three couplings, $\forall\, 1\le u \le p,\ \forall\, 1\le v \le q$:
\begin{equation}
	\pi^+_{u,v}
	= \left[\frac{\mu_u - \mu_1}{q}\right]
	+ \left[\frac{\nu_v - \nu_1}{p}\right]
	+ \left[\frac{\mu_1}{q} + \frac{\nu_1}{p} - \frac{1}{pq}\right].
\end{equation}

Since all three bracketed terms are non-negative (by ordering~\eqref{eq:ordre_marges} and Condition~\eqref{cond:H}), re-normalizing gives:
\begin{equation}\label{eq:decomp1}
	(1 - p\mu_1)\left[\frac{\mu_u - \mu_1}{q(1 - p\mu_1)}\right]
	+ (1 - q\nu_1)\left[\frac{\nu_v - \nu_1}{p(1 - q\nu_1)}\right]
	+ (p\mu_1 + q\nu_1 - 1)\left[\frac{1}{pq}\right].
\end{equation}

\begin{remark}[Tight case: division by zero]
If either of the first two bracketed terms is identically zero, the corresponding marginal $\mu$ or $\nu$ is uniform. In that case, $\pi^+ = \pi^{\times}$, and the interpretation of indeterminacy is trivial.
\end{remark}

Reading~\eqref{eq:decomp1} leads to the decomposition stated in Proposition~\ref{prop:drawing+}.

\begin{proposition}[Sampling under indeterminacy]\label{prop:drawing+}
Assume Condition~\eqref{cond:H}. Let $C\in\{1,2,3\}$ with
\[
\mathbb{P}(C=1)=1-p\mu_1,\qquad \mathbb{P}(C=2)=1-q\nu_1,\qquad \mathbb{P}(C=3)=p\mu_1+q\nu_1-1.
\]
Then a draw $(U,V)\sim \pi^+$ can be generated as follows:
\begin{enumerate}
\item draw $C$;
\item if $C=1$, draw $U$ from $(\mu-\mu_1\mathbf{1})/(1-p\mu_1)$ and $V$ uniformly on $\{1,\dots,q\}$, independently;
\item if $C=2$, draw $U$ uniformly on $\{1,\dots,p\}$ and $V$ from $(\nu-\nu_1\mathbf{1})/(1-q\nu_1)$, independently;
\item if $C=3$, draw $(U,V)$ uniformly on $\{1,\dots,p\}\times\{1,\dots,q\}$.
\end{enumerate}
\end{proposition}

Under this representation, $\pi^+$ removes the largest possible uniform component from each marginal. In particular, following the interpretation of~\cite{harremoes2001inequalities}, indeterminacy can be viewed as the $\mathbb{L}_2$-projection of the uniform distribution on $\{1,\dots,p\}\times\{1,\dots,q\}$ onto $\mathcal{S}_{\mu,\nu}$. In each case of Proposition~\ref{prop:drawing+}, at least one coordinate is drawn uniformly, which minimizes concentration and is consistent with the role of $\pi^+$ in minimizing $IC$.

\section{General Solution Form\label{sec:forme}}
\subsection{Lagrangian Approach}

We rewrite Problem~\ref{pb:max_entropy} in standard form with Lagrange multipliers $\lambda_u$, $\omega_v$, $\theta$, and $r_{u,v}\ge 0$ (values in parentheses).

\begin{problem}\label{pb:chi2_constraints}
\[
\min_{\pi} \sum_{u=1}^p \sum_{v=1}^q \pi_{u,v}^2
\]
subject to the constraints:
\[
\left\{
    \begin{array}{ll}
        \forall u, & \pi_{u,\cdot} = \mu_u \quad (-2\lambda_u),\\
        \forall v, & \pi_{\cdot,v} = \nu_v \quad (-2\omega_v),\\
        \forall (u,v), & -\pi_{u,v} \le 0 \quad (2r_{u,v}),\\
        & \pi_{\cdot,\cdot} = 1 \quad (-2\theta).
    \end{array}
\right.
\]
\end{problem}

Relaxing nonnegativity makes $r$ vanish; Proposition~\ref{prop:solution_pi_plus} then gives $\pi^*=\pi^+$. Thus the difference between the general $\pi^*$ and $\pi^+$ is entirely captured by $r$. The Lagrangian is:

\begin{align*}
L(\pi,r,\lambda,\omega,\theta)
&= \sum_{u,v} \pi_{u,v}^2
- \sum_{u} 2\lambda_u \left(\pi_{u,\cdot} - \mu_u\right)
- \sum_{v} 2\omega_v \left(\pi_{\cdot,v} - \nu_v\right)\\
&\quad - 2\theta \left(\pi_{\cdot,\cdot} - 1\right)
- 2 \sum_{u,v} r_{u,v} \pi_{u,v}.
\end{align*}

\begin{proposition}\label{prop:forme_pi_star}
The solution $\pi^*$ of Problem~\ref{pb:chi2_constraints} takes the following form: for all $(u,v)$,
\[
\pi^*_{u,v}\in\left\{
0,\;
\pi^+_{u,v}
- \frac{R_{u,\cdot}}{q}
- \frac{R_{\cdot,v}}{p}
+ \frac{R}{pq}
\right\},
\]
where, for all $1\le u\le p$ and $1\le v\le q$,
\[
R_{u,\cdot}:=\sum_{v=1}^q r_{u,v},
\qquad
R_{\cdot,v}:=\sum_{u=1}^p r_{u,v},
\qquad
R:=R_{\cdot,\cdot}:=\sum_{u=1}^p\sum_{v=1}^q r_{u,v},
\]
with $r\ge 0$ and $\pi^*\ge 0$.
\end{proposition}
\begin{proof}
A proof is provided in Appendix~\ref{sec:proof_forme_pi_star}.
\end{proof}

\begin{remark}
The correction term $r$ enters through an indeterminacy-type coupling with marginals $R_{u,\cdot}$ and $R_{\cdot,v}$ (Definition~\ref{def:pi_plus}); its relation to $(\mu,\nu)$ is established in Section~\ref{sec:hyp_rectangle}. The presence of null entries in $\pi^*$ is not paradoxical: adding mass to a zero entry forces a compensating reduction elsewhere, which increases $IC$ (see proof of Proposition~\ref{prop:croissance_pi_star}).
\end{remark}

\subsection{A Staircase Structure Containing the Zeros\label{ssec:escalier}}

Under the ordering~\eqref{eq:ordre_marges}, the zero entries of $\pi^*$ form the staircase structure of Figure~\ref{fig:pi_star}. This follows from the following monotonicity property.

\begin{definition}[Monotonicity]\label{def:monotone}
A discrete function $m(u,v)$ of two variables is said to be monotone if, for every fixed $v$ (resp.~$u$), the map $u \mapsto m(u,v)$ (resp.~$v \mapsto m(u,v)$) is non-decreasing.
\end{definition}

\begin{proposition}\label{prop:croissance_pi_star}
When the marginals are ordered according to Equation~\eqref{eq:ordre_marges}, the solution $\pi^*$ is monotone in the sense of Definition~\ref{def:monotone}.
\end{proposition}
\begin{proof}
A proof is provided in Appendix~\ref{sec:proof_croissance_pi_star}.
\end{proof}

\begin{remark}
One can check that Proposition~\ref{prop:croissance_pi_star} holds in particular for $\pi^+$ (Definition~\ref{def:pi_plus}).
\end{remark}

Proposition~\ref{prop:forme_pi_star} implies that $\pi^*$ only takes the values $0$ and strictly positive entries. Row-wise monotonicity forces the zeros in each row to occupy the leftmost columns, while column-wise monotonicity implies that the number of zeros is non-increasing with the row index. Hence $\pi^*$ exhibits the staircase pattern of zeros shown in Figure~\ref{fig:pi_star}.

\begin{figure}[h]
\begin{center}
\begin{tikzpicture}
\matrix[matrix of math nodes,left delimiter={[},right delimiter={]}] (m)
{
0 & \ldots & 0 & 0 & 0 & 0 & \ldots & 0 & +\\
0 & \ldots & 0 & 0 & 0 & + & \ldots & + & +\\
\vdots & \ldots & 0 & 0 & 0 & + & \ldots & + & +\\
0 & \ldots & 0 & + & + & + & \ldots & + & +\\
+ & \ldots & + & + & + & \vdots & \ldots & + & +\\
};

\draw[thick]
(m-5-1.north west) -- 
(m-5-4.north west) -- 
(m-4-4.north west) -- 
(m-4-6.north west) --
(m-2-6.north west) --
(m-2-9.north west) --
(m-1-9.north west) ;

\draw[->,thick]
($(m-1-1.north west)+(0,0.8)$) --
($(m-1-9.north east)+(0,0.8)$)
node[midway,above] {$v$ increases};

\draw[->,thick]
($(m-1-1.north west) + (-0.8,0)$) --
($(m-5-1.south west) + (-0.75,0)$)
node[midway, anchor=center, above, rotate=90] {$u$ increases};

\end{tikzpicture}
\end{center}
\caption{General shape of $\pi^*$: each row begins with a sequence of zeros, followed by strictly positive entries represented by $+$; the number of zeros does not increase when row index $u$ increases.}
\label{fig:pi_star}
\end{figure}

We now formalize the staircase corners.

\begin{definition}[Corners]\label{def:pq_max_null}
For each row $1 \le u \le p$, define the largest null index of $\pi^*$:
\begin{equation}
q_u := \max\left\{1 \le v \le q \; : \; \pi^*_{u,v} = 0 \right\}
\end{equation}
and $q_u:=0$ when the set is empty and row $u$ contains no zeros. Similarly, for each column $1 \le v \le q$, define
\begin{equation}
p_v := \max\left\{1 \le u \le p \; : \; \pi^*_{u,v} = 0 \right\}
\end{equation}
and $p_v:=0$ when the set is empty and column $v$ contains no zeros.
\end{definition}

The monotonicity of $\pi^*$ implies that $q_u$ is non-increasing in $u$ and that $p_v$ is non-increasing in $v$. On each row $u$, the complementary slackness KKT condition $r \pi^* = 0$ for the Lagrange multipliers allows us to write:
\[
\pi^*_{u,v} > 0 \quad \text{ and } \quad r_{u,v} = 0, \quad \forall\, v \ge q_u + 1.
\]

Two elements remain to close the expression of $\pi^*$: the staircase corners $(q_u,p_v)$ and the multipliers $r$.

\section{Rectangle-Solution as a Milestone\label{sec:hyp_rectangle}}
As an analytically tractable intermediate case, we assume the staircase of Figure~\ref{fig:pi_star} collapses to a single upper-left rectangle (Figure~\ref{fig:pi_star_hyp}): this isolates the basic mass-redistribution mechanism and yields an explicit KKT solution. In terms of Definition~\ref{def:pq_max_null}, we assume:
\begin{equation}\label{cond:rectangle}
\exists\, 0\le p_1\le p,\quad \exists\, 0\le q_1\le q,\quad \forall\, 1\le u\le p,\quad q_u = q_1\,\mathds{1}_{u \le p_1}.
\end{equation}

\begin{remark}
We do not address here how Assumption~\eqref{cond:rectangle} translates into conditions on the original marginals; however, Section~\ref{ssec:exemple} provides an example.
\end{remark}

\begin{figure}[h]
\begin{center}
\begin{tikzpicture}
\matrix[matrix of math nodes,left delimiter={[},right delimiter={]}] (m)
{
0 & \ldots & 0 & + & \ldots & +\\
\vdots & \ldots & \vdots & + & \ldots & +\\
0 & \ldots & 0 & + & \ldots & +\\
+ & \ldots & + & + & \ldots & +\\
+ & \ldots & \vdots & + & \ldots & +\\
};

\draw[thick]
(m-4-1.north west) -- 
(m-4-4.north west) -- 
(m-1-4.north west);

\draw[->,thick]
($(m-1-1.north west)+(0,0.8)$) --
($(m-1-6.north east)+(0,0.8)$)
node[midway,above] {$v$ increases};

\draw[->,thick]
($(m-1-1.north west) + (-0.8,0)$) --
($(m-5-1.south west) + (-0.75,0)$)
node[midway, anchor=center, above, rotate=90] {$u$ increases};

\end{tikzpicture}
\end{center}
\caption{Zero pattern of $\pi^*$ under Assumption~\eqref{cond:rectangle}: the first $p_1$ rows contain exactly $q_1$ leading zeros, whereas the remaining rows contain none.}\label{fig:pi_star_hyp}
\end{figure}

\begin{remark}
Setting $p_1=q_1=0$ yields an empty zero block. Since $r\pi^*=0$, $r=0$ and therefore $\pi^*=\pi^+$ using Proposition~\ref{prop:forme_pi_star}. Consequently, Condition~\eqref{cond:H} can be viewed as a degenerate instance of Assumption~\eqref{cond:rectangle} and is therefore strictly stronger.
\end{remark}

Assuming Condition~\eqref{cond:rectangle}, the KKT conditions can be solved in closed form, yielding an explicit solution that we construct in this section.

\begin{lemma}[Form of $\pi^*$ under Assumption~\eqref{cond:rectangle}]
Assume Equation~\eqref{cond:rectangle}, then:
\begin{equation}
\forall (u,v), \quad
\pi^*_{u,v} =
\left\{
    \begin{array}{ll}
        0, & \text{if $u\le p_1$ and $v\le q_1$,}\\[0.2em]
		\pi^+_{u,v} - \dfrac{R_{u,\cdot}}{q} - \dfrac{R_{\cdot,v}}{p} + \dfrac{R}{pq}, & \text{otherwise.}
    \end{array}
\right.
\end{equation}
\end{lemma}

\begin{proof}
The result follows by combining Proposition~\ref{prop:forme_pi_star} and Proposition~\ref{prop:croissance_pi_star} with Definition~\ref{def:pq_max_null} under Assumption~\eqref{cond:rectangle}.
\end{proof}

\subsection{Direct Computation of the Solution}
Under Assumption~\eqref{cond:rectangle}, both $r$ and the corner $(p_1,q_1)$ admit explicit expressions in terms of the marginals.
In subsection~\ref{sssec:masses} we derive the expressions of $R_{u,\cdot}$ and $R_{\cdot,v}$ from the corner $(p_1,q_1)$ and the marginals. Those aggregations correspond to the mass of $\pi^+$ in the upper-left rectangle where $\pi^*$ is zero (Figure~\ref{fig:pi_star_hyp}). They suffice to define a transformation $T: \pi^+\mapsto\pi^*$ where $\pi^+$ is given by Definition~\ref{def:pi_plus}.
In subsection~\ref{sssec:multipliers} we deduce the explicit form of $r$ that will be used in section~\ref{sec:solution_generale}. Eventually in subsection~\ref{sssec:corner} we relate the corner $(p_1,q_1)$ to the marginals.

\subsubsection{Closed-Form of the Rectangle-Solution~\label{sssec:masses}}
For a given matrix $m$ with $p$ rows and $q$ columns, we introduce the following notations for any $1\le u\le p$ (sum up to column $v'$ on row $u$):
\begin{equation}
m[u,:v'] := \sum_{k=1}^{v'} m_{u,k},
\end{equation}
and more generally (sum from column $v$ to $v'$ on row $u$):
\begin{equation}
m[u,v:v'] := \sum_{k=v}^{v'} m_{u,k},
\end{equation}
with the symmetric definitions $m[:u',v]$ and $m[u:u',v]$. We use those notations to express the masses $\Delta$s that are removed when setting the upper-left rectangle $(p_1,q_1)$ of $\pi^+$ to zero. We then show that $\pi^*$ can be derived from $\pi^+$ and those $\Delta$s.

\begin{definition}[Mass loss]\label{def:perte_masse}
For $u\le p_1$ and $v\le q_1$, define the row/column deficits induced by setting to zero the upper-left rectangle $(p_1,q_1)$ of $\pi^+$:
\begin{align}
\Delta_{u,\cdot}
&:= \mu_u - \sum_{v=q_1+1}^q \pi^+_{u,v}
= \pi^+[u,:q_1], \qquad &&1\le u\le p_1, \nonumber\\
\Delta_{\cdot,v}
&:= \nu_v - \sum_{u=p_1+1}^p \pi^+_{u,v}
= \pi^+[:p_1,v], \qquad &&1\le v\le q_1, \nonumber\\
\Delta
&:= \sum_{u=1}^{p_1} \sum_{v=1}^{q_1} \pi^+_{u,v}
= 1 - \sum_{u=1}^p \sum_{v=q_1+1}^q \pi^+_{u,v}
      - \sum_{u=p_1+1}^p \sum_{v=1}^{q_1} \pi^+_{u,v}. \label{def:Delta}
\end{align}
\end{definition}

From the $\Delta$s we explicitly derive the aggregated quantities $R_{u,\cdot}$ and $R_{\cdot,v}$ :

\begin{lemma}\label{lem:calcul_r}
Assume Equation~\eqref{cond:rectangle}. For any row $1\le u\le p$ and any column $1\le v\le q$, we have:
\begin{align}
R_{u,\cdot}
&= -\frac{q}{q-q_1}\left[\Delta_{u,\cdot} + \frac{\Delta}{p-p_1}\right]\mathds{1}_{u\le p_1}, \label{eq:Ru}\\
R_{\cdot,v}
&= -\frac{p}{p-p_1}\left[\Delta_{\cdot,v} + \frac{\Delta}{q-q_1}\right]\mathds{1}_{v\le q_1}, \label{eq:Rv}
\end{align}
Consequently,
\begin{equation}
R= -\frac{pq}{(p-p_1)(q-q_1)}\,\Delta. \label{eq:R}
\end{equation}
\end{lemma}
\begin{proof}
A proof is provided in Appendix~\ref{sec:proof_calcul_r}.
\end{proof}

Following Proposition~\ref{prop:forme_pi_star}, $\pi^*$ only depends on the aggregated quantities of $r$. Therefore we introduce the transformation $T$ together with the rectangle-solution $\pi^* = T(\pi^+)$. 

\begin{theorem}[Rectangle-solution $T(\pi^+)$]\label{th:forme_pi_tilde}
Assume Condition~\eqref{cond:rectangle} then:
\begin{equation}
\begin{split}
\pi^*_{u,v}
&:= \pi^+_{u,v}
+ \frac{\sum_{j=1}^{q_1}\pi^+_{u,j}}{q-q_1}\mathds{1}_{u\le p_1}
+ \frac{\sum_{i=1}^{p_1}\pi^+_{i,v}}{p-p_1}\mathds{1}_{v\le q_1}\\
&\quad + \frac{\sum_{u=1}^{p_1}\sum_{v=1}^{q_1}\pi^+_{u,v}}{(p-p_1)(q-q_1)}
\left[\mathds{1}_{u\le p_1} + \mathds{1}_{v\le q_1} - 1\right],
\end{split}
\end{equation}
for $u\ge p_1+1$ or $v\ge q_1+1$, and $\pi^*_{u,v}=0$ otherwise.
We quote $T:\pi^+\mapsto\pi^*$.
\end{theorem}

\begin{proof}
This result is a direct rewriting of Proposition~\ref{prop:forme_pi_star}, with the incorporation of explicit expressions of aggregations of $r$ given by Lemma~\ref{lem:calcul_r}.
\end{proof}

Under Assumption~\eqref{cond:rectangle}, $\pi^*=T(\pi^+)$ redistributes the negative mass of $\pi^+$ toward larger indices while preserving the marginals. The transformation $T$ is the building block of the general algorithm.

\subsubsection{Lagrange Multipliers $r$\label{sssec:multipliers}}
The characterization above only involves the aggregated quantities $R_{u,\cdot}$ and $R_{\cdot,v}$. For completeness, we now recover an explicit expression for the entrywise multipliers $r_{u,v}$. This formula also serves as the template used to update $r$ in the general algorithm (Algorithm~\ref{alg:one_step_general}).

\begin{proposition}\label{prop:r_uv}
Assume Equation~\eqref{cond:rectangle}. The Lagrange multipliers satisfy
\begin{equation}\label{eq:detail_r}
r_{u,v}
= -\pi^+_{u,v}
  - \frac{\Delta_{u,\cdot}}{q-q_1}
  - \frac{\Delta_{\cdot,v}}{p-p_1}
  - \frac{\Delta}{(p-p_1)(q-q_1)}
\end{equation}
for $1\le u \le p_1,~ 1\le v \le q_1$ and $r_{u,v} = 0$ elsewhere.
\end{proposition}
\begin{proof}
We first check that Equation~\eqref{eq:detail_r} is consistent with Lemma~\ref{lem:calcul_r}: it yields the aggregated quantities $R_{u,\cdot}$ and $R_{\cdot,v}$. 

We now argue that $r$ is uniquely determined. Stationarity of the Lagrangian, Equation~\eqref{eq:gradient_lagrangien}, implies that, on the zero-rectangle $(u\le p_1,\ v\le q_1)$, the multipliers can be written as $r_{u,v} = \lambda_u+\omega_v+\theta$.
Assume that two multipliers $r$ and $r'$ satisfy the KKT system with representations $r_{u,v}=\lambda_u+\omega_v+\theta$ and $r'_{u,v}=\lambda'_u+\omega'_v+\theta'$. Since the row and column aggregates of $r$ are fixed by Lemma~\ref{lem:calcul_r}, it follows that $\omega_v-\omega'_v$ is constant in $v$, and similarly $\lambda_u-\lambda'_u$ is constant in $u$. Hence, $r-r'$ is constant on the rectangle, and the total aggregate constraint Equation~\eqref{eq:R} forces this constant to be zero. Therefore $r=r'$.
\end{proof}

The explicit forms of $r$ and transformation $T$ are key elements of Algorithm~\ref{alg:one_step_general}, which generalizes the mass-redistribution transformation $T$ together with the associated update of $r$. 

\subsubsection{Corner of the Rectangle-Solution\label{sssec:corner}}
Under Assumption~\eqref{cond:rectangle}, the corner $(p_1,q_1)$ can be identified explicitly from the marginals. Focusing on the first row, define
\begin{equation}\label{eq:q_1}
\hat{q}_1 = \min\left\{v'\geq 0 ~\middle|~ \pi^+_{1,v'+1} + \frac{\pi^+[1,:v']}{q-v'} \geq 0 \right\},
\end{equation}
with $\hat{q}_1=0$ if Condition~\eqref{cond:H} holds, and symmetrically define $\hat{p}_1$ from the first column. One can show that $(p_1,q_1)=(\hat{p}_1,\hat{q}_1)$: any alternative corner with $q_1'>\hat{q}_1$ produces a KKT multiplier $r$ violating $r\ge 0$, so $\hat{q}_1$ is the smallest feasible choice; the argument for $\hat{p}_1$ is symmetric. A detailed proof is given in~\cite{bertrand2026arxiv}.

Two properties used in Section~\ref{ssec:analyse_pi_tilde} follow directly:

\begin{lemma}\label{lem:rectangle_negatif}
For all $v > \hat{q}_1$ and all $u$, and for all $u > \hat{p}_1$ and all $v$, $\pi^+_{u,v}\geq 0$.
\end{lemma}
\begin{proof}
Straightforward from the definition of $\hat{p}_1$, $\hat{q}_1$ and the ordering of marginals~\eqref{eq:ordre_marges}.
\end{proof}

\begin{lemma}\label{lem:Delta_négatif}
The total mass $\Delta$ defined in Equation~\eqref{def:Delta} is negative.
\end{lemma}
\begin{proof}
Under the ordering~\eqref{eq:ordre_marges}, $\pi^+$ is monotone (Definition~\ref{def:monotone}), so $\Delta_{u,\cdot}$ is non-decreasing in $u$. Assume by contradiction that $\Delta>0$. Since $\Delta = \sum_{u=1}^{p_1} \Delta_{u,\cdot}$, it follows that $\Delta_{p_1,\cdot}>0$, which implies $R_{p_1,\cdot}<0$ via Equation~\eqref{eq:Ru} a contradiction.
\end{proof}

\paragraph{Testing Rectangle Hypothesis}
Given $(\mu,\nu)$, assume $(p_1,q_1) = (\hat{p}_1,\hat{q}_1)$ and build the candidate coupling $T(\pi^+)$ (Theorem~\ref{th:forme_pi_tilde}) together with the corresponding multipliers $r$ (Proposition~\ref{prop:r_uv}). Whenever $r\geq 0$ and $T(\pi^+)\geq 0$, KKT optimality implies $\pi^*=T(\pi^+)$. Moreover, if KKT conditions are not satisfied for corner $(\hat{p}_1,\hat{q}_1)$, no corner would work and Assumption~\eqref{cond:rectangle} is not satisfied.

Subsection~\ref{ssec:exemple} illustrates a case where Condition~\eqref{cond:H} fails but Assumption~\eqref{cond:rectangle} holds: $T(\pi^+)$ yields an explicit optimizer. Subsection~\ref{ssec:analyse_pi_tilde} provides a block-wise interpretation of the transformation $T$.

\subsection{Structure of the Rectangle-Solution\label{ssec:analyse_pi_tilde}}

The transformation $T$ acts on each of the four blocks induced by corner $(p_1,q_1)$.
\begin{itemize}
\item \emph{Upper-left} ($u\le p_1$, $v\le q_1$): entries are set to zero (the region where $\pi^+$ may be negative, by Lemma~\ref{lem:rectangle_negatif}), removing total mass $\Delta<0$ (Lemma~\ref{lem:Delta_négatif}).
\item \emph{Upper-right} ($u\le p_1$, $v>q_1$): $T(\pi^+)_{u,v}=\pi^+_{u,v}+\Delta_{u,\cdot}/(q-q_1)$, distributing row $u$'s mass deficit $\Delta_{u,\cdot}$ evenly over its $q-q_1$ remaining columns.
\item \emph{Lower-left} ($u>p_1$, $v\le q_1$): $T(\pi^+)_{u,v}=\pi^+_{u,v}+\Delta_{\cdot,v}/(p-p_1)$, symmetrically redistributing column deficits.
\item \emph{Lower-right} ($u>p_1$, $v>q_1$): $T(\pi^+)_{u,v}=\pi^+_{u,v}-\Delta/[(p-p_1)(q-q_1)]$, adding the uniform positive correction $|\Delta|/[(p-p_1)(q-q_1)]$ to compensate the mass shifts above.
\end{itemize}
In all four cases $T$ preserves marginals while eliminating the negative part of $\pi^+$.

\subsection{Rectangle-Solution Example}\label{ssec:exemple}

Consider marginals $\mu=(0.1,0.2,0.3,0.4)$ and $\nu=(0.1,0.3,0.6)$. The corresponding additive coupling is
\[
\pi^+ = \frac{1}{120}
\begin{pmatrix}
-3 & 3 & 12 \\
1 & 7 & 16\\
5 & 11 & 20\\
9 & 15 & 24
\end{pmatrix}.
\]

Condition~\eqref{cond:H} fails since $\pi^+$ has negative entries. Nevertheless, the exact optimizer can still be constructed.
Set, $p_1=\hat{p}_1=1$ and $q_1=\hat{q}_1=1$, so that $\Delta^1_{1,\cdot}=\Delta^1_{\cdot,1}=-\frac{3}{120}$. Zeroing out $\pi^+_{1,1}$ yields $T(\pi^+)$ and the associated multipliers $r$ (Proposition~\ref{prop:r_uv}):
\[
T(\pi^+) = \frac{1}{120}
\begin{pmatrix}
0 & 1.5 & 10.5\\
0 & 7.5 & 16.5\\
4 & 11.5 & 20.5\\
8 & 15.5 & 24.5
\end{pmatrix},
\qquad
r = \frac{1}{120}
\begin{pmatrix}
6 & 0 & 0\\
0 & 0 & 0\\
0 & 0 & 0\\
0 & 0 & 0
\end{pmatrix}.
\]
Since $T(\pi^+)\ge 0$ and $r\ge 0$, we obtain $\pi^*=T(\pi^+)$.
Finally,
\[
IC\big(T(\pi^+)\big)=\frac{1914}{120^2}
\quad\text{whereas}\quad
IC(\pi^+)=\frac{1896}{120^2},
\]
so the latter value is unattainable under the non-negativity constraint.

\section{Solution in the General Case\label{sec:solution_generale}}
We generalize the mass-redistribution $T$ of Theorem~\ref{th:forme_pi_tilde} to a single-row update and iterate it from $\pi^+$; the result terminates in at most $p-1$ steps and recovers $\pi^*$ (Theorem~\ref{th:general}).

To generalize Equation~\eqref{eq:q_1}, for a given row $u$ and a matrix $m$ with $p$ rows and $q$ columns, define
\begin{equation}
I_u(m) = \left\{v' ~\middle|~ m[u,v'+1] + \frac{m[u,:v']}{q-v'} \geq 0 \right\}.
\end{equation}

\begin{proposition}[Definition of $q_u$]\label{prop:def_q}
Assume $m$ is monotone and defines positive marginals, then $I_u(m)$ is an interval of the form $[q_u(m); q-1]$, where $q_u(m):=\min I_u(m)$.
\end{proposition}

\begin{proof}
Fix an arbitrary $1\le u\le p$.
Since marginal $m_{u,\cdot}$ is positive, we immediately have $q-1\in I_u(m)$.
Now assume that $v\in I_u(m)$ and show that any $v'\ge v$ also belongs to $I_u(m)$:
\begin{align*}
&(q-v')\,m[u,v'+1] + m[u,:v']\\
&= (q-v')\,m[u,v'+1] + m[u,:v] + m[u,v+1:v']\\
&\ge (q-v')\,m[u,v+1] + m[u,:v] + (v'-v)m[u,v+1] \quad \text{(by monotonicity of $m$)}\\
&= (q-v)\,m[u,v+1] + m[u,:v] \ge 0 \quad \text{(since $v\in I_u(m)$)}.
\end{align*}
\end{proof}

\begin{lemma}[Monotonicity of $q_u$]\label{lem:croissance_q}
Assume $m$ is monotone. The sequence of indices $q_u(m)$ is non-increasing in $u$.
\end{lemma}

\begin{proof}
Fix $u'\ge u$ and any $v$. Then
\[
(q-v)\,m[u',v+1] + m[u',:v]
\;\ge\;
(q-v)\,m[u,v+1] + m[u,:v].
\]
Hence, if the expression corresponding to row $u$ is non-negative, the one corresponding to row $u'$ is also non-negative.
\end{proof}

\begin{remark}
In particular, take $\pi^+$ (Definition~\ref{def:pi_plus}) a coupling of two marginals ordered according to Equation~\eqref{eq:ordre_marges}, then $q_u(\pi^+)$ is a non-increasing sequence of integers bounded above by $q-1$.
\end{remark}

Algorithm~\ref{alg:one_step_general} implements a one-step generalization of $T$: mass losses $\Delta$s correspond to Definition~\ref{def:perte_masse} where the rectangle is localized on row $l$ up to column $q_l(m)$. Algorithm~\ref{alg:one_step_general} redistributes the negative mass of a single row $l$ toward larger indices as $T$ would. The transformation preserves the marginals and the algorithm simultaneously updates the Lagrange multiplier matrix $r$.

\begin{algorithm}[h]
\caption{\textbf{(Generalization of T):} construction of row $l$ of $\pi^*$; inputs $(m,l,r)$}
\label{alg:one_step_general}
\begin{algorithmic}
\STATE $(p,q) \leftarrow \text{Dimensions}(m)$
\STATE $\tilde{m} \leftarrow m$
\STATE $\tilde{r} \leftarrow r$
\IF{$\min\{m[l,v],~1\le v\le q\} \ge 0$}
\RETURN $\tilde{m},\tilde{r}$
\ELSE
\STATE $q_l \leftarrow q_l(m)$
\STATE \COMMENT{mass loss computation}
\STATE $\Delta^l \leftarrow m[l,:q_l]$
\STATE $\forall (1\le v\le q_l),~ \Delta^l_v \leftarrow m[l,v]$
\STATE \COMMENT{update of $\tilde{m}$}
\STATE $\tilde{m}[l,v] \leftarrow 0 \quad \forall (1\le v\le q_l)$ \COMMENT{upper-left rectangle}
\STATE $\tilde{m}[l,v] \leftarrow \tilde{m}[l,v] + \frac{\Delta^l}{q-q_l} \quad \forall (q_l+1\le v\le q)$ \COMMENT{upper-right rectangle}
\STATE $\tilde{m}[u,v] \leftarrow \tilde{m}[u,v] + \frac{\Delta^l_v}{p-l} \quad \forall (l+1\le u\le p,\; 1\le v\le q_l)$ \COMMENT{lower-left rectangle}
\STATE $\tilde{m}[u,v] \leftarrow \tilde{m}[u,v] - \frac{\Delta^l}{(p-l)(q-q_l)} \quad \forall (l+1\le u\le p,\; q_l+1\le v\le q)$ \COMMENT{lower-right rectangle}
\STATE \COMMENT{update of $\tilde{r}$}
\STATE $\tilde{r}[u,v] \leftarrow r[u,v] - \frac{\Delta^l_v}{p-l} - \frac{\Delta^l}{(p-l)(q-q_l)} \quad \forall (1\le u\le l-1,\; 1\le v\le q_l)$
\STATE $\tilde{r}[l,v] \leftarrow -m[l,v] - \frac{\Delta^l}{q-q_l} - \frac{\Delta^l_v}{p-l} - \frac{\Delta^l}{(p-l)(q-q_l)} \quad \forall (1\le v\le q_l)$
\ENDIF
\RETURN $\tilde{m},\tilde{r}$
\end{algorithmic}
\end{algorithm}

Iterating Algorithm~\ref{alg:one_step_general} over $l=1,\ldots,p-1$ yields Algorithm~\ref{alg:all_steps_general}. The following lemma establishes that each iteration preserves the structural invariants needed by the next.

\begin{lemma}[Invariants of Algorithm~\ref{alg:one_step_general}]
\label{lem:predicats}
Let $(m,r,l)$ be the inputs and $(\tilde{m},\tilde{r})$ the outputs of Algorithm~\ref{alg:one_step_general}. Suppose:

\begin{enumerate}
\item $m$ is non-decreasing in $v$ \label{predicat_croissance_v}
\item $m$ is non-decreasing in $u$ for $u \ge l$ \label{predicat_croissance_u}
\item $m$ has marginals $\mu$ and $\nu$ respecting ordering~\eqref{eq:ordre_marges} \label{predicat_marges+}
\item the first $l-1$ rows of $m$ are non-negative \label{predicat_m+}
\item the upper-left rectangle of corner $(l-1,q_l(m))$ is null as in Figure~\ref{fig:pi_star}: \label{predicat_rectangle}
\[
\forall\, 1 \le u \le l-1,\ \forall\, 1 \le v \le q_l(m), \quad m_{u,v} = 0
\]
\item the sequence $q_u(m)$ is non-increasing for $u\geq l$\label{predicat_qdecroit}
\item the matrix $r$ is constructed progressively, row by row: \label{predicat_r0}
\[
\forall\, l \le u \le p,\ \forall\, v,\quad r_{u,v} = 0
\]
\item the matrix $r$ is non-negative \label{predicat_r+}
\item there exist real numbers $\{\lambda_u\}_{1\le u\le p}$, $\{\omega_v\}_{1\le v\le q}$, and $\theta$ such that for all $(u,v)$: \label{predicat_gradient}
either $r_{u,v}=0$, or $r_{u,v}$ has the form prescribed by Equation~\eqref{eq:gradient_lagrangien}, which nullifies the gradient of the Lagrangian,
\begin{equation}
r_{u,v} = -\lambda_u - \omega_v - \theta,
\end{equation}
and similarly either $m_{u,v}=0$, or
\begin{equation}
m_{u,v} = \lambda_u + \omega_v + \theta.
\end{equation}
moreover, complementary slackness holds:
\[
r_{u,v}\, m_{u,v} = 0.
\]
\end{enumerate}

Then these assumptions are preserved by one iteration of Algorithm~\ref{alg:one_step_general}. More precisely, after replacing $l$ by $l+1$, $m$ by $\tilde{m}$, and $r$ by $\tilde{r}$, the same properties remain satisfied.
\end{lemma}
\begin{proof}
A proof is provided in Appendix~\ref{sec:proof_predicats}. 
\end{proof}

Lemma~\ref{lem:predicats} guarantees that each call to Algorithm~\ref{alg:one_step_general} produces a non-negative row $l$ and a valid row $l$ of $r$, while maintaining the KKT structure needed for subsequent rows. Algorithm~\ref{alg:all_steps_general} iterates from $m=\pi^+$, $r=0$.

\begin{algorithm}
\caption{Iterating Algorithm~\ref{alg:one_step_general} starting from $\pi^+$}
\label{alg:all_steps_general}
\begin{algorithmic} 
\STATE $i\leftarrow 0$
\STATE $\tilde{\pi}^0 \leftarrow \pi^+$
\STATE $\tilde{r}^0 \leftarrow 0$
\WHILE{NOT $\tilde{\pi}^i \geq 0$} 
\STATE $\left(\tilde{\pi}^{i+1},\tilde{r}^{i+1}\right) \leftarrow \text{Algorithm~\ref{alg:one_step_general} applied to }(\tilde{\pi}^i,\tilde{r}^i,i+1)$
\STATE $i\leftarrow i+1$
\ENDWHILE
\RETURN $\tilde{\pi}^i$
\end{algorithmic}
\end{algorithm}

The main result follows.

\begin{theorem}\label{th:general}
Algorithm~\ref{alg:all_steps_general} outputs $\pi^*$ in strictly less than $p$ iterations.
\end{theorem}
\begin{proof}
A proof is provided in Appendix~\ref{sec:proof_general}.
\end{proof}

In summary, iterating a suitable generalization of the transformation $T$ (introduced to build the rectangle-solution of Theorem~\ref{th:forme_pi_tilde}) yields a constructive procedure that computes the exact optimizer of Problem~\ref{pb:max_entropy} in finite time for arbitrary marginals.

\paragraph{Computational complexity.}
A naive implementation costs $O(p^2 q)$: $p-1$ iterations each updating the full $p\times q$ matrix.
With lazy shift tracking --- recording additive corrections to rows rather than materialising them --- each iteration reduces to $O(q)$, giving overall $O(pq)$, matching~\cite{calvillo2016closest}; iterating over the shorter dimension (by symmetry of Problem~\ref{pb:max_entropy}) replaces $p$ by $\min(p,q)$.
By comparison, general-purpose QP solvers yield only numerical approximations of the optimum and incur $O(p^3q^3)$ operations~\cite{wu2026quadratic}.

We now quantify how restrictive Condition~\eqref{cond:H} is, i.e., how often $\pi^*=\pi^+$. This analysis shows that the general construction developed above is necessary in most instances.

\section{How Often Does the Existing Solution Apply?\label{sec:margins}}
We compute, under uniform drawing of marginals over the simplex, the probability that Condition~\eqref{cond:H} holds. We start with the symmetric case $\nu=\mu$, where the condition reduces to $\mu_1 \ge \frac{1}{2p}$.

\begin{proposition}\label{pp:H_log2}
The proportion of $\mu \in S_p$ such that the pair $(\mu,\mu)$ satisfies Condition~\eqref{cond:H} is equal to
\[
\frac{1}{2^{p-1}}.
\]
\end{proposition}

\begin{proof}
Condition~\eqref{cond:H} imposes lower bounds on the coordinates of $\mu$, which in turn restricts the domain of integration defining the uniform measure on $S_p$.
The admissible set has Lebesgue measure
\begin{eqnarray*}
\int_{\frac{1}{2p}}^{1-\frac{p-1}{2p}}
\int_{\frac{1}{2p}}^{1-\frac{p-2}{2p}-x_1}
\cdots
\int_{\frac{1}{2p}}^{1-\frac{1}{2p}-\sum_{i=1}^{p-2} x_i}
\diff x_1 \cdots \diff x_{p-1}.
\end{eqnarray*}

With the successive changes of variables $x_i \leftarrow x_i + \frac{1}{2p}$, the above integral can be written as
\begin{eqnarray*}
&&\int_{0}^{\frac{1}{2}}
\int_{0}^{\frac{1}{2} -x_1}
\cdots
\int_{0}^{\frac{1}{2} - \sum_{i = 1}^{p-2} x_i}
\diff x_1 \cdots \diff x_{p-1} \\[1ex]
&=& \frac{1}{2^{p-1}}
\int_{0}^{1}
\int_{0}^{1 -y_1}
\cdots
\int_{0}^{1 - \sum_{i = 1}^{p-2} y_i}
\diff y_1 \cdots \diff y_{p-1}
= \frac{1}{2^{p-1}}.
\end{eqnarray*}
\end{proof}

\begin{remark}\label{req:constructionH}
Admissible marginals can be constructed explicitly: Condition~\eqref{cond:H} forces $\mu_u\ge\frac{1}{2p}$, so any $\mu$ of the form $\mu_u = \frac{1}{2p}+\frac{r_u}{2}$ with $r\in S_p$ is admissible --- i.e., $\mu$ lives in a $\tfrac{1}{2}$-contraction of $S_p$.
\end{remark}

For the general (independent) case, $\mu$ and $\nu$ are drawn independently and uniformly over $S_p$ and $S_q$. The following proposition gives a structural characterization of admissible pairs.

\begin{proposition}[Construction of eligible marginals]
\label{prop:alpha_discret}
The pair $(\mu,\nu)$ satisfies Condition~\eqref{cond:H} iff $\exists\,\alpha \in [0,1]$ such that, $\forall\, 1 \le u \le p$, $\forall\, 1 \le v \le q$:

\begin{minipage}{0.44\textwidth}
\begin{equation}
\mu_u \ge \frac{\alpha}{p},
\end{equation}
\end{minipage}
\hspace{2ex}
\begin{minipage}{0.44\textwidth}
\begin{equation}
\nu_v \ge \frac{1-\alpha}{q}.
\end{equation}
\end{minipage}
\end{proposition}

\begin{proof}
Assume first that Condition~\eqref{cond:H} holds and define $\alpha = p \mu_1 \in [0,1]$, where $\mu_1 = \min_u \mu_u$ as in Equation~\eqref{eq:ordre_marges}. Then, for all $1 \le v \le q$,
\[
\nu_v \ge \frac{1-\alpha}{q}.
\]
Conversely, if such $\alpha$ exists, then for all $1 \le u \le p$ and $1 \le v \le q$,
\[
\frac{\mu_u}{q} + \frac{\nu_v}{p}
\ge \frac{\alpha}{pq} + \frac{1-\alpha}{pq}
= \frac{1}{pq},
\]
which is exactly Condition~\eqref{cond:H}.
\end{proof}

\begin{remark}
The factor $p$ in $\alpha=p\mu_1$ preserves the symmetry between $\mu$ and $\nu$, so all $\alpha\in[0,1]$ are admissible regardless of $(p,q)$.
\end{remark}
Generalizing Remark~\ref{req:constructionH}, Proposition~\ref{prop:alpha_discret} yields probability distributions $r\in S_p$ and $s\in S_q$ such that, $\forall\, 1 \le u \le p$ and $\forall\, 1 \le v \le q$:

\begin{minipage}{0.44\textwidth}
\begin{equation}
\mu_u = \frac{\alpha}{p} + (1-\alpha) r_u,
\label{eq:decompr}
\end{equation}
\end{minipage}
\hspace{2ex}
\begin{minipage}{0.44\textwidth}
\begin{equation}
\nu_v = \frac{1-\alpha}{q} + \alpha s_v.
\label{eq:decomps}
\end{equation}
\end{minipage}

\begin{proposition}[Constructive characterization of eligible marginals]\label{prop:propmunu}
A pair of probability laws $(\mu,\nu)\in S_p\times S_q$ satisfies Condition~\eqref{cond:H} if and only if there exist a real $\alpha\in[0,1]$ and a pair of probability laws $(r,s)\in S_p\times S_q$ such that both Equations~\eqref{eq:decompr} and~\eqref{eq:decomps} hold.
\end{proposition}
For fixed $\alpha$, admissible $\mu$ form a $(1-\alpha)$-contraction of $S_p$ and admissible $\nu$ an $\alpha$-contraction of $S_q$. Since $\mu$ and $\nu$ are drawn independently, the admissible measure on $S_p\times S_q$ is
\begin{equation}\label{eq:H_proportion}
\int_{0}^{1} \alpha^{p-1}(1-\alpha)^{q-1}\,\diff\alpha
= \frac{(p-1)!(q-1)!}{(p+q-2)!}.
\end{equation}
Moreover, these representations completely characterize the pairs of marginals satisfying Condition~\eqref{cond:H}.

\begin{remark}[Different shapes]
The expression of Proposition~\ref{prop:propmunu} does not reduce to that of Equation~\eqref{eq:H_proportion} by simply setting $p=q$. This discrepancy arises because independence between marginals holds only in the second setting.
\end{remark}

\paragraph{Scaling with dimension.}
Both expressions above show that the feasibility Condition~\eqref{cond:H} becomes increasingly unlikely as the problem size grows. In the symmetric case $(\mu,\mu)$, Proposition~\ref{pp:H_log2} yields a proportion $2^{-(p-1)}$, which decays exponentially in $p$. In the independent case $(\mu,\nu)$, the proportion $\frac{(p-1)!(q-1)!}{(p+q-2)!}$ also decreases as either $p$ or $q$ increases, reflecting the fact that satisfying lower bounds simultaneously on all coordinates becomes harder in higher dimensions.

\section{Conclusions}
We studied the problem of minimizing the $\mathbb{L}_2$-norm over the transportation polytope $\mathcal{S}_{\mu,\nu}$, equivalently maximizing Rényi entropy of order~$2$ under marginal constraints.
Three results were established.

First, the feasibility condition under which the relaxed optimizer $\pi^+$ is itself optimal (Condition~\eqref{cond:H}) holds for a proportion $(p-1)!(q-1)!/(p+q-2)!$ of marginal pairs under the uniform distribution on $S_p\times S_q$; this proportion decays exponentially in $p$ and $q$, establishing that the general algorithm is required in all but degenerate cases.

Second, the general optimizer $\pi^*$ admits a complete KKT characterization with explicit closed-form Lagrange multipliers (Theorem~\ref{th:forme_pi_tilde}, Proposition~\ref{prop:r_uv}), complementing~\cite{romero1990easy,calvillo2016closest}, whose geometric approach does not call for Lagrange multipliers.

Third, the indeterminacy coupling $\pi^+$ decomposes into a mixture of three interpretable couplings and coincides with the $\mathbb{L}_2$-projection of the uniform distribution onto $\mathcal{S}_{\mu,\nu}$, yielding a closed-form sampling procedure.

Open questions include extending the framework to more than two marginals, deriving a closed-form sampler for $\pi^*$ in the general case, and studying continuous analogues~\cite{bubak2025fitting}.
Beyond its theoretical interest, the algorithm applies naturally wherever avoiding bottlenecks in a network is a primary objective: pedestrian and freight routing, cryptanalysis (where the index of coincidence measures exploitable non-uniformity), and risk allocation (where regulators seek to avoid concentration of exposure on individual counterparties).

\section*{Acknowledgements}
The project leading to this publication has received funding from the French government under the ``France 2030'' investment plan managed by the French National Research Agency (reference: ANR-17-EURE-0020) and from Excellence Initiative of Aix-Marseille University -- A*MIDEX.

The author thanks Michel Broniatowski and Jean-Fran\c{c}ois Marcotorchino for their precious help and fruitful discussions.

\bibliographystyle{plain}
\bibliography{Biblio}

\appendix

\makeatletter
\renewcommand\paragraph{\@startsection{paragraph}{4}{\z@}%
  {1.5ex \@plus .5ex \@minus .2ex}%
  {0.5em}%
  {\normalfont\normalsize\bfseries}}
\def\subparagraph{\@startsection{subparagraph}{5}{\parindent}%
  {0.75ex \@plus .2ex \@minus .1ex}%
  {0.3em}%
  {\normalfont\normalsize\itshape}}
\makeatother

\section{Proof of Proposition~\ref{prop:forme_pi_star}\label{sec:proof_forme_pi_star}}
\begin{proof}
Setting $\partial L/\partial\pi_{u,v}=0$ gives:
\begin{equation}\label{eq:gradient_lagrangien}
\pi^*_{u,v} = \lambda_u + \omega_v + \theta + r_{u,v};
\end{equation}
the complementary slackness conditions imply that, for all $(u,v)$,
\begin{equation}
\pi^*_{u,v} \, r_{u,v} = 0;
\end{equation}
and we know $r\geq 0$.
Summing Equation~\ref{eq:gradient_lagrangien} respectively over $u$ and $v$, using self-explanatory notation ($\Lambda := \sum_{u=1}^p \lambda_u$, $\Omega := \sum_{v=1}^q \omega_v$), yields:
\begin{equation}\label{eq:gradient_lagrangien_sum}
\left\{
    \begin{array}{ll}
        \forall u, & \mu_u = q \lambda_u + \Omega + q \theta + R_{u,\cdot},\\
		\forall v, & \nu_v = \Lambda + p \omega_v + p \theta + R_{\cdot,v}.
    \end{array}
\right.
\end{equation}
Summing both equations:
\begin{equation}\label{eq:theta}
 \theta = \frac{1}{pq}
 - \frac{\Lambda}{p}
 - \frac{\Omega}{q}
 - \frac{R}{pq}.
\end{equation}
Rewriting Equation~\eqref{eq:gradient_lagrangien_sum}, we obtain:
\begin{equation}
\left\{
    \begin{array}{ll}
        \forall u, & \lambda_u
        = \dfrac{\mu_u - \Omega - q\theta - R_{u,\cdot}}{q},\\
		\forall v, & \omega_v
        = \dfrac{\nu_v - \Lambda - p\theta - R_{\cdot,v}}{p}.
    \end{array}
\right.
\end{equation}
Hence:
\begin{equation}
\pi^*_{u,v}
= \frac{\mu_u - R_{u,\cdot}}{q}
+ \frac{\nu_v - R_{\cdot,v}}{p}
- \frac{\Omega}{q}
- \frac{\Lambda}{p}
- \theta
+ r_{u,v}.
\end{equation}
Substituting~\eqref{eq:theta}:
\begin{equation}
\pi^*_{u,v}
= \pi^+_{u,v}
- \frac{R_{u,\cdot}}{q}
- \frac{R_{\cdot,v}}{p}
+ \frac{R}{pq}
+ r_{u,v}.
\end{equation}
This recovers the $\pi^+$ form with $r$-correction terms. Complementary slackness $r_{u,v}\pi^*_{u,v}=0$ completes the proof.
\end{proof}

\section{Proof of Proposition~\ref{prop:croissance_pi_star}\label{sec:proof_croissance_pi_star}}
\begin{proof}
By contradiction; by symmetry it suffices to treat fixed $u$. Suppose $\pi^*_{u_0,v_0} > \pi^*_{u_0+1,v_0}$ for some $u_0 < p$ and $v_0$. Since $\mu_{u_0}\le\mu_{u_0+1}$ and $\pi^*\ge 0$, compensation is required: there exists $v_0'$ with
\[
\pi^*_{u_0,v_0'} < \pi^*_{u_0+1,v_0'}.
\]
Define
\[
\epsilon
= \frac{1}{2}
\min\!\left(
\pi^*_{u_0,v_0} - \pi^*_{u_0+1,v_0},
\;
\pi^*_{u_0+1,v_0'} - \pi^*_{u_0,v_0'}
\right)
> 0.
\]
We then perturb $\pi^*$ by setting $\pi^{**}$ equal to $\pi^*$ everywhere except at four entries:
\begin{eqnarray*}
\pi^{**}_{u_0,v_0} &=& \pi^*_{u_0,v_0} - \epsilon,\\
\pi^{**}_{u_0,v_0'} &=& \pi^*_{u_0,v_0'} + \epsilon,\\
\pi^{**}_{u_0+1,v_0} &=& \pi^*_{u_0+1,v_0} + \epsilon,\\
\pi^{**}_{u_0+1,v_0'} &=& \pi^*_{u_0+1,v_0'} - \epsilon.
\end{eqnarray*}
By construction $\pi^{**}\in\mathcal{S}_{\mu,\nu}$, and $\epsilon$ is chosen so that $\pi^{**}\ge 0$. Computing the objective change:
\begin{eqnarray*}
\sum_{u,v} (\pi^{**}_{u,v})^2 - \sum_{u,v} (\pi^*_{u,v})^2
&=& \left(\pi^*_{u_0,v_0} - \epsilon\right)^2 - (\pi^*_{u_0,v_0})^2
+ \left(\pi^*_{u_0,v_0'} + \epsilon\right)^2 - (\pi^*_{u_0,v_0'})^2 \\
&+& \left(\pi^*_{u_0+1,v_0} + \epsilon\right)^2 - (\pi^*_{u_0+1,v_0})^2
+ \left(\pi^*_{u_0+1,v_0'} - \epsilon\right)^2 - (\pi^*_{u_0+1,v_0'})^2 \\
&=& 4\epsilon^2
+ 2\epsilon\!\left(\pi^*_{u_0+1,v_0} - \pi^*_{u_0,v_0}\right)
+ 2\epsilon\!\left(\pi^*_{u_0,v_0'} - \pi^*_{u_0+1,v_0'}\right) \\
&\le& -4\epsilon^2 < 0.
\end{eqnarray*}

Thus $\pi^{**}\in\mathcal{S}_{\mu,\nu}$ achieves a strictly smaller value of $IC$ than $\pi^*$, contradicting optimality.
\end{proof}

\section{Proof of Lemma~\ref{lem:calcul_r}\label{sec:proof_calcul_r}}
\begin{proof}
We first prove Equation~\eqref{eq:R} by summing over all non-zero entries of $\pi^*$:
\begin{eqnarray*}
1
&=& \sum_{u=1}^p \sum_{v=q_1+1}^q \pi^*_{u,v}
   + \sum_{u=p_1+1}^p \sum_{v=1}^{q_1} \pi^*_{u,v} \\
&=& \sum_{u=1}^p \sum_{v=q_1+1}^q
\left(\pi^+_{u,v} - \frac{R_{u,\cdot}}{q} - \frac{R_{\cdot,v}}{p} + \frac{R}{pq}\right)  
+ \sum_{u=p_1+1}^p \sum_{v=1}^{q_1}
\left(\pi^+_{u,v} - \frac{R_{u,\cdot}}{q} - \frac{R_{\cdot,v}}{p} + \frac{R}{pq}\right).
\end{eqnarray*}

Since $r$ vanishes on the support of $\pi^*$ (so $R_{u,\cdot}=0$ for $u\ge p_1+1$), Definition~\ref{def:perte_masse} gives:
\begin{equation}
\Delta = -(q-q_1)\frac{R}{q} -R + (q-q_1)\frac{R}{q} - 0 -(p-p_1)\frac{R}{p}+q_1(p-p_1)\frac{R}{pq}=-\frac{(q-q_1)(p-p_1)}{pq}R.
\end{equation}

For row $u\le p_1$ (otherwise $r_{u,\cdot}=0$):
\begin{eqnarray*}
\mu_u
&=& \sum_{v=q_1+1}^q \pi^*_{u,v}
= \sum_{v=q_1+1}^q
\left(\pi^+_{u,v} - \frac{R_{u,\cdot}}{q} - \frac{R_{\cdot,v}}{p} + \frac{R}{pq}\right) \\
&=& \sum_{v=q_1+1}^q \pi^+_{u,v}
 - \frac{(q-q_1)R_{u,\cdot}}{q}
 + \frac{(q-q_1)R}{pq}.
\end{eqnarray*}

By Definition~\ref{def:perte_masse} and~\eqref{eq:R}:
\begin{equation}
\Delta_{u,\cdot}
= -\frac{(q-q_1)R_{u,\cdot}}{q}
  - \frac{\Delta}{p-p_1}.
\end{equation}

For $u\ge p_1+1$, $R_{u,\cdot}=0$ by KKT, justifying $\mathds{1}_{u\le p_1}$. Equation~\eqref{eq:Rv} follows by symmetry.
\end{proof}

\section{Proof of Lemma~\ref{lem:predicats}\label{sec:proof_predicats}}
\begin{proof}
We establish each invariant separately.

\paragraph*{Invariant~\ref{predicat_croissance_v}: monotonicity of $\tilde{m}$ in $v$}
Fix $u<l$. Then $v \mapsto \tilde{m}_{u,v} = m_{u,v}$ is non-decreasing by assumption~\ref{predicat_croissance_v}.

When $u=l$: $\tilde{m}_{l,v}=0$ for $v\le q_l$ (non-decreasing), and $\tilde{m}_{l,v}=m_{l,v}+\Delta^l/(q-q_l)$ for $v\ge q_l+1$ (non-decreasing since a constant is added to a non-decreasing function). The junction $\tilde{m}_{l,q_l+1}\ge 0=\tilde{m}_{l,q_l}$ holds by construction.

When $u>l$: for $v\le q_l$, we add $m_{l,v}/(p-l)$ (non-decreasing by Assumption~\ref{predicat_croissance_v}) to $m_{u,v}$ (non-decreasing); for $v\ge q_l+1$, we add the constant $\Delta^l/[(p-l)(q-q_l)]$. It remains to verify $\tilde{m}_{u,q_l}\le\tilde{m}_{u,q_l+1}$.
Since $q_l-1 \notin I_l(m)$, we have:
\begin{eqnarray}
&& m_{l,q_l} \le - \frac{m[l,:q_l-1]}{q-q_l+1}\nonumber\\
&\Leftrightarrow & m_{l,q_l}(q-q_l+1) \le -(\Delta^l - m_{l,q_l})\nonumber\\
&\Leftrightarrow & m_{l,q_l}(q-q_l) \le -\Delta^l\nonumber\\
&\Leftrightarrow & \Delta^l_{q_l} \le - \frac{\Delta^l}{q-q_l}.
\label{eq:majoration_ql}
\end{eqnarray}
By Assumption~\ref{predicat_croissance_v},
\[
\tilde{m}_{u,q_l} \le \tilde{m}_{u,q_l+1}.
\]

\paragraph{Invariant~\ref{predicat_croissance_u}: monotonicity of $\tilde{m}$ in $u$ for $u \ge l+1$}
For $u\ge l+1$ the update adds only a column-dependent constant; monotonicity in $u$ (Assumption~\ref{predicat_croissance_u}) is preserved.

\begin{remark}
Monotonicity between rows $l$ and $l+1$ is not guaranteed, since for $v \le q_l$ one has $\tilde{m}_{l,v}=0$ and it may occur that $\tilde{m}_{l+1,v} < 0$.
\end{remark}

\paragraph*{Invariant~\ref{predicat_marges+}: preservation of the marginals by $\tilde{m}$}
Rows $u<l$ are unchanged. For the remaining cases:

For row $l$:
\begin{eqnarray*}
\tilde{m}_{l,\cdot}
&=& \sum_{v=q_l+1}^q \tilde{m}_{l,v}
= \sum_{v=q_l+1}^q \left(m_{l,v} + \frac{\Delta^l}{q-q_l}\right) \\
&=& \sum_{v=q_l+1}^q m_{l,v} + \Delta^l
= m_{l,\cdot}.
\end{eqnarray*}

For rows $u>l$:
\begin{eqnarray*}
\tilde{m}_{u,\cdot}
&=& \sum_{v=1}^q \tilde{m}_{u,v} 
= \sum_{v=1}^{q_l}\left(m_{u,v} + \frac{\Delta^l_v}{p-l}\right)
+ \sum_{v=q_l+1}^q\left(m_{u,v} - \frac{\Delta^l}{(p-l)(q-q_l)}\right)\\
&=& \sum_{v=1}^q m_{u,v}
+ \frac{\Delta^l}{p-l}
- \frac{\Delta^l}{p-l}
= m_{u,\cdot}.
\end{eqnarray*}

For columns $v \le q_l$:
\begin{eqnarray*}
\tilde{m}_{\cdot,v}
&=& \sum_{u=l+1}^p \tilde{m}_{u,v}
= \sum_{u=l+1}^p\left(m_{u,v} + \frac{\Delta^l_v}{p-l}\right)\\
&=& \sum_{u=l+1}^p m_{u,v} + \Delta^l_v
= \sum_{u=l}^p m_{u,v}
= m_{\cdot,v},
\end{eqnarray*}
where the last equality follows from Assumption~\ref{predicat_rectangle}.

For columns $v \ge q_l$:
\begin{eqnarray*}
\tilde{m}_{\cdot,v}
&=& \sum_{u=1}^p \tilde{m}_{u,v}
= \tilde{m}[1:l-1,v] + \tilde{m}_{l,v} + \tilde{m}[l+1:p,v]\\
&=& m[1:l-1,v] + m[l,v] + \frac{\Delta^l}{q-q_l}
+ m[l+1:p,v] - (p-l)\frac{\Delta^l}{(p-l)(q-q_l)}\\
&=& m_{\cdot,v}.
\end{eqnarray*}

\paragraph*{Invariant~\ref{predicat_m+}: non-negativity of $\tilde{m}$ up to row $l$}
Rows $u<l$ are non-negative by Assumption~\ref{predicat_m+}. Row $l$ is zero for $v\le q_l$ and positive at $v=q_l+1$ by construction.

\paragraph*{Invariant~\ref{predicat_rectangle}: zero rectangle below $q_{l+1}(\tilde{m})$}
We show $q_u(\tilde{m})=q_u(m)$ for all $u>l$; existence of $q_u(\tilde{m})$ follows from the monotonicity already established. For $u>l$, since $q_u(m)\le q_l(m)$ (Assumption~\ref{predicat_qdecroit}), any $v'\le q_u(m)$ satisfies $v'\le q_l(m)$, so the explicit form of $\tilde{m}$ applies:
\begin{eqnarray*}
&&(q-v')\tilde{m}_{u,v'} + \tilde{m}[u,:v'-1] \\
&=& (q-v')\left(m_{u,v'} + \frac{m_{l,v'}}{p-l}\right)
+ m[u,:v'-1] + \frac{m[l,:v'-1]}{p-l}\\
&=& (q-v')m_{u,v'} + m[u,:v'-1]
+ \frac{1}{p-l}\left((q-v')m_{l,v'} + m[l,:v'-1]\right).
\end{eqnarray*}
This is non-negative iff $v'\ge q_u(m)$, so
\[
q_u(\tilde{m}) = q_u(m).
\]
In particular,
\[
q_{l+1}(\tilde{m}) = q_{l+1}(m) \le q_l(m)
\]
by Assumption~\ref{predicat_qdecroit}. Hence, using Assumption~\ref{predicat_rectangle},
\[
\forall\, 1\le u\le l-1,\; 1\le v\le q_{l+1}(\tilde{m}),\quad \tilde{m}_{u,v}=m_{u,v}=0.
\]
For row $l$: since $q_{l+1}(\tilde{m})=q_{l+1}(m)\le q_l(m)$ and row $l$ is zero up to $q_l(m)$,
\[
\tilde{m}_{l,v}=0 \quad \text{for all } v\le q_{l+1}(\tilde{m}).
\]

\paragraph*{Invariant~\ref{predicat_qdecroit}: monotonicity of $q_u(\tilde{m})$ from row $l+1$ onward}
Since $q_u(\tilde{m})=q_u(m)$ for $u>l$, the result follows from Assumption~\ref{predicat_qdecroit}.

\paragraph*{Invariant~\ref{predicat_r0}: $\tilde{r}$ null after row $l+1$}
The algorithm only modifies rows $\le l$; by Assumption~\ref{predicat_r0}, $r$ is null after row $l$, so $\tilde{r}$ is null after row $l+1$.

\paragraph*{Invariant~\ref{predicat_r+}: non-negativity of $\tilde{r}$}
By monotonicity of $m$ in $v$, for all $v\le q_l$ we have
\[
\Delta^l_v = m_{l,v} \le m_{l,q_l} = \Delta^l_{q_l}.
\]
Using inequality~\eqref{eq:majoration_ql}, this implies
\[
\Delta^l_v + \frac{\Delta^l}{q-q_l} \le 0.
\]
For $v\le q_l$, the update adds either $-\left(\Delta^l_v+\frac{\Delta^l}{q-q_l}\right)$ or $-\frac{1}{p-l}\left(\Delta^l_v+\frac{\Delta^l}{q-q_l}\right)$ to $r_{u,v}\ge 0$; both quantities are non-negative by~\eqref{eq:majoration_ql}. For $v>q_l$, $\tilde{r}_{u,v}=r_{u,v}\ge 0$. Hence $\tilde{r}\ge 0$.

\paragraph*{Invariant~\ref{predicat_gradient}: preservation of the Lagrangian stationarity structure}
We set by default $\tilde{\lambda}=\lambda$, $\tilde{\omega}=\omega$, and $\tilde{\theta}=\theta$, and modify:
\begin{itemize}
\item $\forall\, u\le l-1,\quad \tilde{\lambda}_u = \lambda_u + \frac{\Delta^l}{(p-l)(q-q_l)}$,
\item $\tilde{\lambda}_l = \lambda_l + \frac{\Delta^l}{q-q_l} + \frac{\Delta^l}{(p-l)(q-q_l)}$,
\item $\forall\, v\le q_l,\quad \tilde{\omega}_v = \omega_v + \frac{\Delta^l_v}{p-l} + \frac{\Delta^l}{(p-l)(q-q_l)}$,
\item $\tilde{\theta} = \theta - \frac{\Delta^l}{(p-l)(q-q_l)}$.
\end{itemize}
We then verify that the resulting forms of $\tilde{r}$ and $\tilde{m}$ match the KKT stationarity conditions.

\subparagraph*{Case $1\le u\le l-1,\; 1\le v\le q_l$}
Whenever $\tilde{r}$ is non-zero:
\begin{eqnarray*}
\tilde{r}_{u,v}
&=& r_{u,v} - \frac{\Delta^l_v}{p-l} - \frac{\Delta^l}{(p-l)(q-q_l)} \\
&=& -\lambda_u - \omega_v - \frac{\Delta^l_v}{p-l}
- \frac{\Delta^l}{(p-l)(q-q_l)} - \theta\\
&=& -\tilde{\lambda}_u - \tilde{\omega}_v - \tilde{\theta}.
\end{eqnarray*}
Here $\tilde{m}=0$ by Assumption~\ref{predicat_rectangle}.

\subparagraph*{Case $u=l,\; 1\le v\le q_l$}
Here $\tilde{m}=0$ by construction. Whenever $\tilde{r}$ is non-zero:
\begin{eqnarray*}
\tilde{r}_{l,v}
&=& - m_{l,v} - \frac{\Delta^l_v}{p-l}
- \frac{\Delta^l}{q-q_l}
- \frac{\Delta^l}{(p-l)(q-q_l)}\\
&=& -\lambda_l - \omega_v - \theta
- \frac{\Delta^l}{q-q_l}
- \frac{\Delta^l_v}{p-l}
- \frac{\Delta^l}{(p-l)(q-q_l)}\\
&=& -\tilde{\lambda}_l - \tilde{\omega}_v - \tilde{\theta}.
\end{eqnarray*}

\subparagraph*{Case $u=l,\; q_l+1\le v\le q$}
Here $\tilde{r}=0$ (unchanged from $r=0$ by Assumption~\ref{predicat_r0}). Whenever $\tilde{m}$ is non-zero:
\begin{eqnarray*}
\tilde{m}_{l,v}
&=& m_{l,v} + \frac{\Delta^l}{q-q_l}
= \lambda_l + \omega_v + \theta + \frac{\Delta^l}{q-q_l}
= \tilde{\lambda}_l + \tilde{\omega}_v + \tilde{\theta}.
\end{eqnarray*}

\subparagraph*{Case $u>l,\; 1\le v \le q_l$}
Here $\tilde{r}=0$ and
\[
\tilde{m}_{u,v}
= m_{u,v} + \frac{\Delta^l_v}{p-l}
= \lambda_u + \omega_v + \frac{\Delta^l_v}{p-l} + \theta
= \tilde{\lambda}_u + \omega_v + \frac{\Delta^l_v}{p-l} + \theta
= \tilde{\lambda}_u + \tilde{\omega}_v + \tilde{\theta}.
\]

\subparagraph*{Case $u>l,\; q_l+1\le v \le q$}
Here $\tilde{r}=0$ and
\[
\tilde{m}_{u,v}
= m_{u,v} - \frac{\Delta^l}{(p-l)(q-q_l)}
= \tilde{\lambda}_u + \tilde{\omega}_v + \theta
- \frac{\Delta^l}{(p-l)(q-q_l)}
= \tilde{\lambda}_u + \tilde{\omega}_v + \tilde{\theta}.
\]

In all cases, $(\tilde{r},\tilde{m})$ satisfies complementary slackness and the stationarity condition~\eqref{eq:gradient_lagrangien} up to row $l$.
\end{proof}

\section{Proof of Theorem~\ref{th:general}\label{sec:proof_general}}
\begin{proof}
We verify: (i) initialization satisfies Lemma~\ref{lem:predicats}, (ii) Algorithm~\ref{alg:all_steps_general} terminates, (iii) the output equals~$\pi^*$.
\paragraph*{Initialization}
At $i=0$: Assumptions~\ref{predicat_croissance_v}--\ref{predicat_croissance_u} hold since $\pi^+$ is monotone (Definition~\ref{def:monotone}); Assumption~\ref{predicat_marges+} holds by~\eqref{eq:ordre_marges}; Assumptions~\ref{predicat_m+} and~\ref{predicat_rectangle} are vacuous; Assumption~\ref{predicat_qdecroit} follows from Lemma~\ref{lem:croissance_q}; and Assumptions~\ref{predicat_r0}--\ref{predicat_r+} hold since $r=0$.
Defining
\[
\lambda_u = \frac{\mu_u}{q}, \qquad
\omega_v = \frac{\nu_v}{p}, \qquad
\theta = \frac{1}{pq}
\]
verifies Assumption~\ref{predicat_gradient}.

\paragraph*{Algorithm terminates}
If Algorithm~\ref{alg:one_step_general} exits early at iteration $i\le p-2$, row $l=i+1$ is non-negative; Assumption~\ref{predicat_croissance_u} then gives $m\ge 0$ and the loop halts.

Otherwise the loop reaches $i=p-2$ ($l=p-1$) and performs an update. Row $p-1$ has a negative entry, so by Assumption~\ref{predicat_croissance_v}, $m_{p-1,1}<0$ and $q_l(m)\ge 2$. Assumptions~\ref{predicat_marges+} and~\ref{predicat_rectangle} then yield $\tilde{m}_{p,1}=\nu_1>0$, and Assumption~\ref{predicat_croissance_v} makes the entire last row of $\tilde{m}$ positive. Assumption~\ref{predicat_m+} (Lemma~\ref{lem:predicats}) ensures the first $p-1$ rows are non-negative. Thus the algorithm terminates at some $i\le p-1$ with $\tilde{m}\ge 0$.

\paragraph*{Output satisfies KKT}
Let $n$ be the termination index and $\tilde{\pi}=\tilde{\pi}^n$. Assumption~\ref{predicat_marges+} gives correct marginals; non-negativity holds by construction. Assumption~\ref{predicat_gradient} gives KKT stationarity and complementary slackness; $r\ge 0$ by Assumption~\ref{predicat_r+}. By Proposition~\ref{prop:kkt}, $\pi^*=\tilde{\pi}$.
\end{proof}

\end{document}